\documentclass[pra,reprint,amssymb,superscriptaddress]{revtex4-1}

\usepackage{natbib}

\usepackage{amsthm}
\usepackage{amsmath}
\usepackage{color}
\usepackage{bm}
\usepackage{amsmath,amssymb}
\usepackage{amsfonts}
\usepackage{dsfont}
\usepackage{graphicx}
\usepackage{float}
\usepackage{subfigure}
\usepackage{verbatim}
\usepackage{amsfonts}
\usepackage{bbold}
\usepackage{dcolumn}
\usepackage{bm}
\usepackage[dvipsnames]{xcolor}
\usepackage{listings}
\definecolor{blueprl}{RGB}{46,48,146}
\usepackage[pdftex,colorlinks=true,urlcolor=blueprl,citecolor=blueprl,linkcolor=blueprl]{hyperref}
\usepackage{mathrsfs}
\usepackage{filecontents}
\usepackage{mathtools}
\usepackage{times}
\usepackage{color}
\usepackage[dvipsnames]{xcolor}

\newtheorem{proposition}{Proposition}

\newcommand{\bra}[1]{\mbox{$\langle #1 |$}}
\newcommand{\ket}[1]{\mbox{$| #1 \rangle$}}

\newcommand{\tr}{\mbox{Tr}}

\begin{document}

\title{Multipartite Gaussian Entanglement of Formation}

\author{Sho Onoe}
\email{sho.onoe@uqconnect.edu.au}
\affiliation{Centre for Quantum Computation and Communication Technology, School of Mathematics and Physics, The University of Queensland, St. Lucia, Queensland, 4072, Australia}
\author{Spyros Tserkis}
\affiliation{Centre for Quantum Computation and Communication Technology, Department of Quantum Science, Australian National University, Canberra, ACT 2601, Australia.}
\author{Austin P. Lund}
\author{Timothy C. Ralph}
\affiliation{Centre for Quantum Computation and Communication Technology, School of Mathematics and Physics, The University of Queensland, St. Lucia, Queensland, 4072, Australia}

\begin{abstract}
Entanglement of formation is a fundamental measure that quantifies the entanglement of bipartite quantum states. This measure has recently been extended into multipartite states taking the name $\alpha$-entanglement of formation. In this work, we follow an analogous multipartite extension for the Gaussian version of entanglement of formation, and focusing on the the finest partition of a multipartite Gaussian state we show this measure is fully additive and computable for 3-mode Gaussian states.
\end{abstract}

\maketitle

\section{Introduction}

Entanglement is a property of quantum mechanics that allows correlations beyond the classical limit. As such, it is considered a crucial resource that allows certain quantum protocols to be more efficient than their classical counterpart \cite{horodecki2009quantum}. Several entanglement measures have been defined in the literature \cite{horodecki2009quantum,plenio2014introduction}, however in general the quantification of their values is a challenging task.

Bipartite entanglement of formation (EoF) \cite{bennett1996mixed} is defined as the least expected amount of bipartite entropy of entanglement (EoE) required to create a state. In general, the quantification of bipartite EoF involves a minimization over infinite degrees of freedom, making it hard to compute. Initial research focused on simple systems such as the 2-qubit system \cite{hill1997entanglement,wootters1998entanglement}, which led to analytical expressions for the measure. 

An analogous measure, called Gaussian EoF (GEoF), focusing only on Gaussian states and operations, was defined by Wolf et al. \cite{wolf2004gaussian}. A few years later, this measure was proven to be equal to EoF in the case of 2-mode Gaussian states \cite{akbari2015entanglement,Wilde.PRA.18}. For these types of states, several efficient numerical methods and analytical expressions have been derived  \cite{marian2008entanglement,wolf2004gaussian,Tserkis.Ralph.PRA.17,tserkis2019quantifying}. Recently, in Ref.~\cite{szalay2015multipartite}, Szalay introduced a measure referred to as $\alpha$-EoF, which is the multipartite extension of bipartite EoF. In this paper, we follow Wolf's approach and apply the notion of $\alpha$-EoF onto the Gaussian regime. We show that $\alpha$-GEoF is a computable multipartite entanglement measure. We utilize a special case of $\alpha$-GEoF, which we refer to as N-mode GEoF, to quantify the \textit{total entanglement} in a 3-mode Gaussian system, in the sense that it is the sum of the entanglement of all internal partitions of the state.

Our paper is set out in the following way. In section II, we introduce the conventions adopted in this paper. In section III we review bipartite entanglement measures. In section IV, we review $\alpha$-entanglement measures \cite{szalay2015multipartite} and introduce a special subset, referring to it as N-mode EoF. In section V, we apply $\alpha$-entanglement measures to the Gaussian regime and prove N-mode EoF is fully additive. In section VI, we consider the tripartite case and compute the total entanglement for simple cases. We summarize and conclude our results in section VII.
\section{Preliminaries}

\subsection{Modes, Partitions and Sub-Systems}

In the discrete variable case, the smallest sub-systems are referred to as qudits (or qubits for 2-level systems). In the continuous variable case, the smallest sub-systems are referred to as modes. For simplicity, this paper will be utilizing the terminology mode, but in this context it can be used interchangeably with qudits if we are not considering the case of Gaussian states. 

Let us consider an $N$-mode state $\hat{\rho}$. The state of the $n$th mode $\hat{\rho}_n$ can be found via the partial trace over all other modes:
\begin{equation}
\hat{\rho}_{n}\equiv \tr_{\forall i\neq n}(\hat{\rho}) \,.
\end{equation}
$\hat{\rho}$ can be split into $M$ partitions, via assigning each mode into one of the $M$ partitions (where $N \geqslant M$). By doing this, we introduce $M$ sub-systems, denoted $\{s_1,s_2,...,s_M\}$. This defines the M-partitioning, $\alpha=s_1|s_2|...|s_M$. Each sub-system $s_j$ is defined as the reduced state, achieved through the partial trace over all other sub-system, i.e.,
\begin{equation}
\hat{\rho}_{s_n}\equiv \tr_{\forall s_i\neq s_n}(\hat{\rho}) \,.
\end{equation}

\subsection{von Neumann Entropy}

Before we get into the quantification of entanglement, we need to first define a function that a broad family of entanglement measures are based on, i.e., quantum entropy \cite{Neumann1927, bravyi2003entanglement, ohya2004quantum}. In particular, we focus on the von Neumann entropy, which for a state $\hat{\rho}$ is defined as
\begin{equation}
S(\hat{\rho}) \equiv - \tr (\hat{\rho} \ln \hat{\rho}) \,.
\end{equation}
$S(\hat{\rho})$ is a symmetric, basis-independent function, which vanishes for pure states. Also, note that it is fully additive for non-correlated states (although sub-additive in general), i.e.,
\begin{equation}
S\left(\hat{\rho}_{s_1} \otimes \hat{\rho}_{s_2} \right) =S(\hat{\rho}_{s_1})+S(\hat{\rho}_{s_2}) \,,
\end{equation}
and convex
\begin{equation}
S \left(\sum_j p_j \hat{\rho}_j \right) \geqslant \sum_j p_j S(\hat{\rho}_j) \,.
\end{equation}

\subsection{Gaussian States}

In the later part of this paper, we will be considering quantum systems comprised of bosonic Gaussian {modes}, $\hat{a}_n$ \cite{welsch1999ii, weedbrook2012gaussian, adesso2014continuous, serafini2017quantum}. These bosonic annihilation operators satisfy the bosonic commutation relations $[\hat{a}_n,\hat{a}_m^{\dag}]=\delta^n_m$, where $\delta$ is the Kronecker delta. For Gaussian states, the analysis of first and second moment \cite{weedbrook2012gaussian} is sufficient to characterize the Wigner function of a particular output mode \cite{scully1999quantum}. The first moment of an $N$-mode Gaussian state is fully characterized by its $2N$-dimensional displacement vector, $\vec{D}$. The second order moment is described by its $2N \times 2N$ real symmetric covariance matrix \cite{wang2007quantum}, $\boldsymbol{\sigma}$. As a result, all Gaussian states can be written as $\hat{\rho}_{\boldsymbol{\sigma},\vec{D}}$.

The $i$th element of the displacement vector is defined in the following way:
\begin{gather}
\vec{D}_{i} \equiv \tr (\rho \hat{R}_i) \,,
\end{gather}
where
\begin{gather}
\vec{R} \equiv(\hat{q}_{1},...,\hat{q}_{N},\hat{p}_{1},...,\hat{p}_{N})^{T} \,,
\end{gather}
and we have defined $\hat{q}_n\equiv \hat{a}_n+\hat{a}_n^{\dag}$ and $\hat{p}_n\equiv \hat{a}_n-\hat{a}_n^{\dag}$. The $\{i,i'\}$th element of the covariance matrix $\boldsymbol{\sigma}$ is defined in the following way:
\begin{equation}
\sigma_{ii'} \equiv \tr \left[\hat{\rho}(\hat{R}_i\hat{R}_{i'}+\hat{R}_{i'}\hat{R}_i)\right]-2\tr \left[\hat{\rho} \hat{R}_i\right] \tr \left[\hat{\rho} \hat{R}_{i'}\right] \,.
\end{equation}

\section{Bipartite Entanglement Measures}

\subsection{Bipartite Entropy of Entanglement}

EoE, $E_{s_1|s_2}$, is the typical way to quantify bipartite entanglement in pure states, $\hat{\psi} := \ket{\psi}\bra{\psi}$ \cite{bennett1996concentrating}. This measure is given by the von Neumann entropy of the reduced state:
\begin{equation}
E_{s_1|s_2}(\hat{\psi}) \equiv S\left[\tr_{s_2}(\hat{\psi})\right]\,.
\end{equation}
As $\hat{\psi}$ is a pure state, EoE is invariant under permutations, i.e., $E_{s_1|s_2}(\hat{\psi}) = E_{s_2|s_1}(\hat{\psi})$. This is a reliable bipartite entanglement measure as it satisfies the following postulates \cite{horodecki2009quantum, plenio2014introduction}:
\begin{enumerate}
	\item $E_{s_1|s_2}$ is an indicator function for separability between the subsystem $s_1$ and $s_2$;
\begin{equation}
E_{s_1|s_2}(\hat{\psi})=0 \Leftrightarrow \hat{\psi}=\hat{\psi}_{s_1} \otimes \hat{\psi}_{s_2} \,.
\end{equation}
 \item $E_{s_1|s_2}$ is non-increasing on average under local operations and classical communications (LOCC), $\hat{\Lambda}_{s_1|s_2}$, where the locality is defined in terms of the sub-system $s_1$ and $s_2$ \cite{bennett1996mixed, vidal2000entanglement, horodecki2001separability, chitambar2014everything, eltschka2014quantifying, plenio2014introduction};
\begin{equation}
E_{s_1|s_2}(\hat{\psi}) \geqslant \sum_j p_j E_{s_1|s_2}\left[\hat{\Lambda}_{j,s_1|s_2} (\hat{\psi})\right],
\end{equation} 
where
\begin{equation}
\hat{\Lambda}_{s_1|s_2} (\hat{\psi}) = \sum_j p_j \hat{\Lambda}_{j,s_1|s_2} (\hat{\psi}) \, ,
\end{equation}
are pure LOCC sub-operations \cite{horodecki2005simplifying, szalay2015multipartite}.

\end{enumerate}

\subsection{Bipartite Entanglement of Formation}

A natural way to extend an entanglement measure to mixed states is via the convex-roof extension \cite{uhlmann2010roofs, bennett1996mixed, uhlmann2000fidelity, rothlisberger2009numerical}. EoF is defined as the convex-roof extension of EoE:
\begin{equation}
\mathcal{E}_{\text{F},s_1|s_2}(\hat{\rho}) \equiv \inf_{\hat{\rho}=\sum_j p_j \hat{\psi}_{j}}\left\{\sum_j p_j E_{s_1|s_2}(\hat{\psi}_{j}) \right\} \,,
\end{equation}
where ``$\inf$'' becomes a ``$\min$'' for discrete variable states, and the sum can be replaced with an integral when considering a continuum of pure states.

This is a reliable bipartite entanglement measure as it satisfies the mixed state extension of the aforementioned postulates \cite{szalay2015multipartite} and an extra one, i.e.,
\begin{enumerate} \setcounter{enumi}{2}
	\item for pure states $\mathcal{E}_{\text{F},s_1|s_2}$ reduces to the entropy of entanglement, i.e.,
\begin{equation}
\mathcal{E}_{\text{F},s_1|s_2}(\hat{\psi})={E}_{s_1|s_2}(\hat{\psi}) \,.
\end{equation}
\end{enumerate}
As von Neumann entropy is convex, postulate 2 implies that bipartite EoF is also non-increasing under LOCC; $\mathcal{E}_{\text{F},s_1|s_2}(\hat{\rho})\geq \mathcal{E}_{\text{F},s_1|s_2}\left[\hat{\Lambda}_{s_1|s_2}(\hat{\rho})\right]$.
\section{M-Partite Entanglement Measures}

\subsection{$\alpha$ - Separability}

Entanglement can also exist among several partitions. There are several ways to divide an $N$-mode system into $M$ partitions. To make a distinction between the partitioning, Szalay \cite{szalay2015multipartite} introduced a hierarchy of separability classes. A pure state, $\ket{\psi}_{\alpha}$, is called ``$\alpha$-separable'' when
\begin{equation}
\ket{\psi}_{\alpha} \equiv \bigotimes_{s_i \in \alpha}\ket{\psi_{s_i}} \,.
\end{equation}

For example, a pure five-mode state is $1|23|45$-separable if and only if the state can be written in the following way
\begin{equation}
\ket{\psi}_{1|23|45} = \ket{\psi_{1}} \otimes \ket{\psi_{23}}\otimes \ket{\psi_{45}} \,.
\end{equation}

Then an $\alpha$-separable mixed state can be written in the following way
\begin{equation}
\hat{\rho}_{\alpha}=\sum_j p_j \ket{\psi_j}_{\alpha}\bra{\psi_j}_{\alpha} \,.
\end{equation}
We can then make a hierarchy for separability as follows: $\alpha$ precedes or equals $\beta$, if all sub-system in $\beta$ can be written as a subset or equal to a subsystem in $\alpha$, i.e.,
\begin{equation}
\alpha \preceq \beta \Leftrightarrow \forall s_i \in \beta ,  \; \exists s_{i'} \in \alpha : s_i \subseteq s_{i'} \,.
\end{equation}
If $\alpha$ has a finer partition than $\beta$ (i.e. $\alpha \preceq \beta$), then a state which is $\beta$ separable must also be $\alpha$ separable. 

\subsection{$\alpha$-Entropy of Entanglement and $\alpha$-Entanglement of Formation}

\subsubsection{$\alpha$-Von Neumann Entropy}

Let us define $\alpha$-von Neumann entropy in the following way:
\begin{equation}
S_{\alpha}(\hat{\rho})\equiv \frac{1}{2}\sum_{s_i\in \alpha}S(\hat{\rho}_{s_i}) \,.
\end{equation}

This is a measure that is well-defined for all states $\hat{\rho}$. Due to the fully additivity of $S$, $S_{\alpha}$ is also fully additive:
\begin{equation}
S_{\alpha}(\hat{\rho}_A\otimes \hat{\rho}_B)= S_{\alpha_A}(\hat{\rho}_A) +S_{\alpha_B}(\hat{\rho}_B) \,,
\end{equation}
where $\alpha_C$, $C \in \{A,B\}$, is the subset of $\alpha$ which includes the part that overlaps with the system $C$.

\subsubsection{$\alpha$-Entropy of Entanglement and Entanglement of Formation}

In the multipartite case, Szalay \cite{szalay2015multipartite} defined the $\alpha$-EoE of a pure state $\hat{\psi}$ to be:
\begin{equation}
E_{\alpha}(\hat{\psi})= S_{\alpha}(\hat{\psi}) \,.
\end{equation}
This measure can be interpreted as the sum of entanglement between the partitions. 

$\alpha$-EoF is defined as the convex-roof extension to $\alpha$-EoE\cite{szalay2015multipartite}:
\begin{equation}
\mathcal{E}_{\text{F},\alpha}(\hat{\rho}) \equiv \inf_{\hat{\rho}=\sum_j p_j \hat{\psi}_{j}}\left\{\sum_j p_j E_{\alpha}(\hat{\psi}_{j}) \right\} \,.
\end{equation}

$\alpha$-EoE and EoF are reliable $\alpha$-entanglement measure as they satisfy the same postulates as the bipartite case, except we must replace $s_1|s_2$ with $\alpha$. $\alpha$-entanglement measures also satisfy an extra postulate:
\begin{enumerate} \setcounter{enumi}{3}
	\item $E_\alpha$ and $\mathcal{E}_\alpha$ must satisfy the multipartite monotonicity;
\begin{gather}
{E}_\alpha(\hat{\rho})\leqslant {E}_\beta(\hat{\rho}), \; \forall \alpha \preceq \beta\,,
\\
\mathcal{E}_{\text{F},\alpha}(\hat{\rho})\leqslant \mathcal{E}_{\text{F},\beta}(\hat{\rho}), \; \forall \alpha \preceq \beta \,.
\end{gather}
\end{enumerate}
This means that an entanglement measure of finer partition is sensitive to more entanglement within the system, hence giving a larger value.
\subsection{N-Mode Entropy of Entanglement and N-Mode Entanglement of Formation}
\subsubsection{N-Mode Entropy of Entanglement}
In this section, we consider the finest partitioning of $\alpha$-entanglement measure and refer to it as the $N$-mode entropy of entanglement (NEoE) and formation (NEoF).  We notice that we have replaced the term partition with mode, as we are no longer interested in the entanglement between the partition that we assign, but with every mode that exists within the system, i.e. $N=M$. NEoE and NEoF satisfy the same postulates as $\alpha$-entanglement measures with the finest partitioning. 

For a pure $N$-mode state, $\hat{\psi}$, NEoE is defined in the following way:
\begin{equation}
\tilde{E}(\hat{\psi})= \tilde{S}(\hat{\psi}) \equiv \frac{1}{2}\sum_{n=1}^{\mathrm{N}} \; S\left[Tr_{\forall i\neq n}(\hat{\psi})\right] \,.
\end{equation} 
NEoE is the sum of all entanglement between each mode and the rest of the system. Due to multipartite monotonicity, this measure is also the largest pure entanglement measure out of the $\alpha$-EoF. For this reason, we refer to this quantity as the total of entanglement within the system.

To highlight a feature of this measure, let us consider a 2-mode entangled state, with a vacuum input in the 3rd mode. In this case, this measure will reduce down to the bipartite entanglement between the 2-mode entangled state, giving the total entanglement within the system. In comparison, a genuine tripartite entanglement measure \cite{adesso2006multipartite, schneeloch2020quantifying} would be zero in this case, as there is only bipartite entanglement.
\subsubsection{N-Mode Entanglement of Formation}
For an $N$-mode mixed state, $\hat{\rho}$, NEoF is defined in the following way:
\begin{equation}
\tilde{\mathcal{E}}_{\text{F}}(\hat{\rho}) \equiv \inf_{\hat{\rho}=\sum_j p_j \hat{\psi}_{j}}\left\{\sum_{j} {p_j} \tilde{E}(\hat{\psi}_{j}) \right\} \,.
\end{equation}

This measure quantifies the least expected total entanglement that is required to create the mixed state. Even though this is a well-defined measure it is hard to compute as there are infinite degrees of freedom for the set $\{p_j,\hat{\psi}_{j}\}$. In this paper, we limit ourselves to a Gaussian convex roof-extension to overcome this problem.
\section{$\alpha$-Gaussian Entanglement of Formation}

\subsection{Von Neumann Entropy and $\alpha$-EoE for Gaussian States}

For Gaussian states, the von Neumann entropy of a state, $\hat{\rho}_{\boldsymbol{\sigma},\vec{D}}$, is fully characterized by its covariance matrix. The von Neumann entropy of an $N$-mode Gaussian state with covariance matrix $\boldsymbol{\sigma}$ can be calculated as follows \cite{Agarwak1971}:
\begin{equation}
	S(\boldsymbol{\sigma})=\frac{1}{2}\sum_{n=1}^{N} \; h(\nu_n) \,,
\end{equation}
where $\nu_n$ is the $n$th symplectic eigenvalue of $\boldsymbol{\sigma}$, and
\begin{equation}
	h(x) \equiv \frac{x_+}{2}\log_2(\frac{x_+}{2})-\frac{x_-}{2}\log_2(\frac{x_-}{2}) \,,
\end{equation}
with $x_{\pm} \equiv x\pm 1$ an auxiliary function.

As the von Neumann entropy is fully characterized by its covariance matrix, $\alpha$-EoE of a pure state, $\hat{\psi}_{\boldsymbol{\pi},\vec{D}}$, is also fully characterized by its covariance matrix. The $\alpha$-EoE of a pure state with covariance matrix $\boldsymbol{\pi}$ is calculated as follows:
\begin{equation}
E_{{\alpha}}(\boldsymbol{\pi}) =\frac{1}{2} \sum_{s_i\in \alpha} S\left[\tr_{s_i}(\boldsymbol{\pi})\right] \,.
\end{equation}
A covariance matrix is pure if and only if $\text{det}(\boldsymbol{\pi})=1$.
%We now denote the covariance matrix of a pure Gaussian state as $\boldsymbol{\boldsymbol{\pi}}$. 

\subsection{$\alpha$-Gaussian Entanglement of Formation}

A mixed Gaussian state $\rho_{\boldsymbol{\sigma},\vec{D}}$ can be decomposed into a mixture of pure Gaussian states in the following way:
\begin{equation}
\hat{\rho}_{\boldsymbol{\sigma},\vec{D}}=\int  \mathrm{d} \boldsymbol{\pi} \mathrm{d}\vec{D}' \; \mu (\boldsymbol{\pi},\vec{D}') \hat{\psi}_{\boldsymbol{\pi},\vec{D}'} \,,
\end{equation}
where $\mu$ is the probability density of $\hat{\rho}_{\boldsymbol{\pi},\vec{D}'}$. In Ref.~\cite{wolf2004gaussian} the authors defined the bipartite Gaussian entanglement of formation (GEoF), and analogously we define the $\alpha$-GEoF as follows
\begin{equation}\label{GEoF}
\begin{aligned}
\mathcal{E}_{\text{G},{\alpha}}(\hat{\rho}_{\boldsymbol{\sigma},\vec{D}}) \equiv \underset{\mu}{\text{inf}}\left\{ \int\right. & \mathrm{d} \boldsymbol{\pi} \mathrm{d}\vec{D}' \; \mu (\boldsymbol{\pi},\vec{D}')E_\alpha(\boldsymbol{\pi}) 
\\
&  \left. | \hat{\rho}_{\boldsymbol{\sigma},\vec{D}}=\int  \mathrm{d} \boldsymbol{\pi} \mathrm{d}\vec{D}' \; \mu (\boldsymbol{\pi},\vec{D}') \hat{\psi}_{\boldsymbol{\pi},\vec{D}'} \right\} \, .
\end{aligned}
\end{equation}
This definition involves a minimization over infinite degrees of freedom, however by following the analysis of Ref.~\cite{wolf2004gaussian}, we find that Eqn. (\ref{GEoF}) reduces to the following expression
\begin{equation}\label{Wolf}
\mathcal{E}_{\text{G},{\alpha}}(\boldsymbol{\sigma})=\underset{\pi}{\text{inf}}\left\{ E_{\alpha}(\boldsymbol{\pi})| \boldsymbol{\sigma}=\boldsymbol{\pi} + \boldsymbol{\varphi} \right\} \,,
\end{equation}
where $\boldsymbol{\varphi}$ is a positive semi-definite matrix. 
%This NGEoF measure has the virtue of finding the optimum pure state (i.e. the pure state with the lowest N-mode entanglement entropy) to deterministically create the state via Gaussian LOCC.
This equation has finite free parameters, and therefore is a computable entanglement measure. In the App. \ref{ApAdd} we utilize Eqn. (\ref{Wolf}) to prove the additivity of NGEoF.
\section{N-mode Gaussian Entanglement of Formation for 3 mode states}

\subsection{Mixed 3-mode Gaussian states}

For mixed 3-modes states, we can utilize Gaussian local unitary operations (GLUO; refer to App. \ref{ApGLUO}) to reduce the state into the standard form \cite{adesso2006multipartite, ferraro2005gaussian}:
\begin{equation}
\boldsymbol{\sigma}_{\text{sf}}=
\begin{bmatrix}
a_1 &e_{1} &e_{3} &0 & 0 & e_{4}\\
e_{1}&  a_2 &e_{6} &0 &0 &  e_{7}\\
e_{3}&  e_{6} & a_{3} & 0 & e_8  &0 \\
0 &0 &0 & a_1  & e_{2} & e_{5}\\
0& 0 & e_{8}& e_{2} & a_2  & e_{9} \\
e_4& e_7&0 &e_{5}&  e_{9}  & a_{3} 
\end{bmatrix} \,.
\end{equation}

As GLUO do not affect the entanglement, we can reduce Eqn. (\ref{Wolf}) to the following:
 \begin{equation}\label{EqnGEoF}
\mathcal{E}_{\text{G},\alpha}(\boldsymbol{\sigma})=\underset{\boldsymbol{\pi}}{\text{inf}}\left\{ E(\boldsymbol{\pi})| \boldsymbol{\sigma}_{\text{sf}}-\boldsymbol{\pi}\geqslant 0 \right\} \,.
\end{equation}
In the next subsection, we fully parametrize $\boldsymbol{\pi}$.

\subsection{Pure 3-mode Gaussian States}

By utilizing GLUO, $\textbf{L}$, we can reduce any $\boldsymbol{\pi}$ to the standard form \cite{duan2000inseparability, simon2000peres}. For the 3-mode pure state, the standard form is \cite{adesso2005gaussian}:
\begin{equation}
\boldsymbol{\pi}_{\text{sf}}
=\begin{bmatrix}
a_1 & e_{12}^{+}& e_{13}^{+}&0  &0  &0\\
e_{12}^{+}& a_2& e_{23}^{+} & 0  &0  & 0\\
e_{13}^{+}& e_{23}^{+}& a_{3} & 0  &0 & 0 \\
0&0&0 & a_1  & e_{12}^{-} & e_{13}^{-}\\
0&0&0 &e_{12}^{-} & a_2  & e_{23}^{-} \\
0&0&0 & e_{13}^{-} & e_{23}^{-}  & a_{3} 
\end{bmatrix} \,,
\end{equation}
where $e_{ij}^{\pm}$ are a function of $a_1$, $a_2$ and $a_3$. For $\boldsymbol{\pi}_{\text{sf}}$ to be a physical covariance matrix the inequality	$|a_i-a_j|\leqslant a_k-1$ must be satisfied \cite{adesso2006multipartite}. All pure states can then be decomposed in the following way:
\begin{equation}
\begin{aligned}
\boldsymbol{\pi} = \boldsymbol{L} \boldsymbol{\pi}_{\text{sf}}(a_1,a_2,a_3) \boldsymbol{L}^T \,.
\end{aligned}
\end{equation}

In general, $\boldsymbol{L}$ has 9 free parameters, and hence the minimization of Eqn. (\ref{EqnGEoF}) can be conducted over 12 free parameters. A numerical code which scans over all possible $\boldsymbol{\pi}$ with finite size step for these 12 free parameters can be created. The condition $\left(\boldsymbol{\sigma}_{\text{sf},n}-\boldsymbol{\pi}_n \right) \geqslant 0 $, gives a finite range for all local squeezing operations, $a_1$, $a_2$ and $a_3$. The phase parameters are limited to $0 \geqslant \phi \geqslant 2 \pi $.

\subsection{q-p states}
In this section we consider a special class of states where we can reduce the number of free parameters to 6. In special cases, the standard form of the mixed state reduces to the following form:
\begin{equation}
\boldsymbol{\sigma}_{\text{qp}}=
\begin{bmatrix}
a_1 & e_{1}& e_{3}&0  &0  & 0\\
e_{1}& a_2& e_{6} & 0  &0  & 0\\
e_{3}& e_{6}&a_{3} & 0  & 0   &0 \\
0 &0 &0 & a_1 & e_{2} & e_{5}\\
0& 0& 0&e_{2} & a_2  & e_{9} \\
0& 0&0& e_{5} & e_{9}  & a_{3} 
\end{bmatrix} \,.
\end{equation}
We will refer to these states as q-p states. q-p states have the property that the $\hat{q}$-quadrature is completely uncorrelated to the $\hat{p}$-quadrature. This means that we can write the following:
\begin{equation}
\boldsymbol{\sigma}_{\text{qp}}=\boldsymbol{\sigma}_{\hat{q}} \oplus \boldsymbol{\sigma}_{\hat{p}} \,.
\end{equation}
Following the analysis in Ref.~\cite{wolf2004gaussian}, we prove in App. \ref{ApQp} that the optimum pure state to create such a state must also be a q-p state:
\begin{equation}
\boldsymbol{\pi}_{\text{qp}}=\boldsymbol{\pi}_{\hat{q}} \oplus \boldsymbol{\pi}_{\hat{p}} \,.
\end{equation}
These states only have 6 free parameters, which greatly reduces the complexity of the problem.

\subsection{Numerical Results}

Consider a two-mode Gaussian state, where one of the modes is thermal, while the others are vacuum. When a two-mode squeezer is applied to such a state, the bipartite GEoF is constant regardless of the number of photon in the thermal mode \cite{adesso2004extremal,adesso2005gaussian,giovannetti2014ultimate}. 
\begin{figure}\label{FigNGEoF}
{\includegraphics[width=\columnwidth]{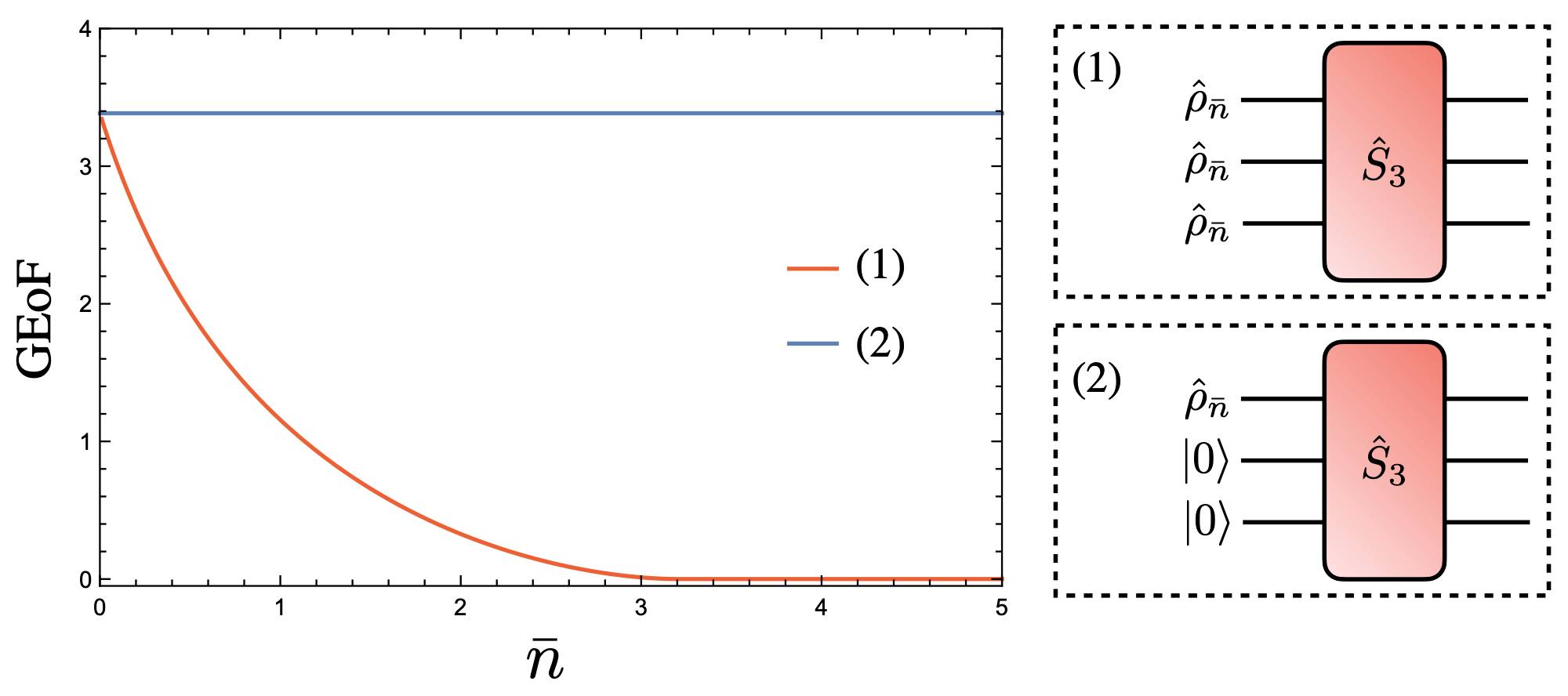}}
\caption{A plot demonstrating how NGEoF changes with input noise. The red line represents the NGEoF for a thermal input,$\hat{rho}_{\bar{n}}$, in all three modes. The blue line represents the NGEoF for thermal input in one mode, with all other modes being a vacuum input. }
\end{figure}

We aim to replicate an analogous result in the tripartite case, utilizing NGEoF. We consider a case where a 3-mode squeezer, $\hat{S}_3$ (details of this operation can be found in App. \ref{ApS3}), is applied to an input with with all three modes which are thermal with an average of $\bar{n}$ particles. Since the output state is a q-p state, we conduct an numerical optimization over q-p state to obtain Fig. 1. We repeat this process in the case where $\hat{S}_3$ is applied to an input with 1 mode which is thermal and the rest being a vacuum. NGEoF is constant in when there is only one thermal input, which is an analogous result to the 2-mode case. 

\section{Conclusion}

In this paper, we utilized the analysis of Ref.~\cite{szalay2015multipartite} on multipartite entanglement measures, and applied it to the Gaussian regime. We successfully demonstrated that the degree of freedom for this measure reduces down to a finite one for all Gaussian states. In particular, we were interested in a special case of $\alpha$-GEoF; NGEoF which quantifies the least expected total entanglement that is required to create the state. We proved that this measure is fully additive. In the last section we quantified its value for simple 3-mode Gaussian states and demonstrated that this measure displayed analogous features to the 2-mode case. 

An interesting future research direction would be to compare NGEoF and NEoF. For the 2-mode case, it has been proven that NGEoF and NEoF coincides with each other for Gaussian states \cite{akbari2015entanglement,Wilde.PRA.18}. It would be beneficial to prove that this can be extended to the N-mode case. Combined with the result that NGEoF is additive, as proven by this paper, the additivity of NEoF would then be proven for Gaussian states in general.

In this paper, we were particularly interested in NEoF, however there are other interesting $\alpha$-EoF measures. In particular, there is a $\alpha$-EoF which quantifies the genuine tripartite entanglement within a three mode system \cite{szalay2015multipartite, schneeloch2020quantifying}. We refer to this measure as a genuine tripartite entanglement measure, as it vanishes for all states which are not genuinely tripartite entangled states. A recent paper \cite{schneeloch2020quantifying} looked into finding an upper bound to this measure for the DV case. It would be interesting to apply this to the Gaussian regime, and investigate how useful the measure is.

%In the 2-mode case, the entanglement measure was a monotone over one parameter. This made the interpretation of entanglement very simple. For the 3-mode case, the entanglement comes in 3 forms; EPR, EPR-type and GhZW states \cite{ferraro2005gaussian}. It would be beneficial to develop an entanglement measure from a resource perspective, to allow the quantification of various entanglement resource.

\section{Acknowledgements}

This work is supported by the Australian Research Council (ARC) under the Centre of Excellence for Quantum Computation and Communication Technology (Grant No. CE170100012).

\onecolumngrid

\appendix

\section{Additivity of N-mode Gaussian Entanglement of Formation}\label{ApAdd}

\begin{proposition}
NGEoF for Gaussian states $\boldsymbol{\sigma}=\boldsymbol{\sigma}_{A}\oplus \boldsymbol{\sigma}_B$ is fully additive, i.e.,
\begin{equation}
\tilde{\mathcal{E}}_{\mathrm{G}}(\boldsymbol{\sigma}_A \oplus \boldsymbol{\sigma}_B)=\tilde{\mathcal{E}}_{\mathrm{G}}(\boldsymbol{\sigma}_A)+\tilde{\mathcal{E}}_{\mathrm{G}}(\boldsymbol{\sigma}_B) \,.
\end{equation}
where $\boldsymbol{\sigma}_A$ and $\boldsymbol{\sigma}_B$ is an $N$-mode and $N'$-mode Gaussian state, respectively.
\end{proposition}

\begin{proof}
NGEoF is by construction sub-additive, i.e.,
\begin{equation}
\tilde{\mathcal{E}}_{G}(\boldsymbol{\sigma}_A \oplus \boldsymbol{\sigma}_B) \leqslant \tilde{\mathcal{E}}_{\mathrm{G}}(\boldsymbol{\sigma}_A)+\tilde{\mathcal{E}}_{\mathrm{G}}(\boldsymbol{\sigma}_B)\,,
\end{equation}
and thus its additivity can be shown by proving that NGEoF is super-additive too, i.e.,
\begin{equation}
\tilde{\mathcal{E}}_{\mathrm{G}}(\boldsymbol{\sigma}_A \oplus \boldsymbol{\sigma}_B) \geqslant \tilde{\mathcal{E}}_{\mathrm{G}}(\boldsymbol{\sigma}_A)+\tilde{\mathcal{E}}_{\mathrm{G}}(\boldsymbol{\sigma}_B) \,.
\label{superadditive}
\end{equation}
The Gaussian state $\boldsymbol{\sigma}=\boldsymbol{\sigma}_A \oplus \boldsymbol{\sigma}_B$ can decomposed as
\begin{equation}
\boldsymbol{\sigma}=\boldsymbol{\sigma}_A \oplus \boldsymbol{\sigma}_B=\boldsymbol{\pi}+\boldsymbol{\varphi} \,,
\label{decomposition}
\end{equation}
where $\boldsymbol{\pi}$ is a pure Gaussian state and $\boldsymbol{\varphi}$ is a positive semidefinite matrix. For any $\boldsymbol{\varphi} \geqslant 0$, the NGEoF for the states $\boldsymbol{\sigma}_A$ and $\boldsymbol{\sigma}_B$ satisfies
\begin{subequations}
\begin{gather}
\tilde{\mathcal{E}}_{\mathrm{G}}[\tr_B(\boldsymbol{\pi})] \geqslant \tilde{\mathcal{E}}_{\mathrm{G}}[\tr_B (\boldsymbol{\pi}+\boldsymbol{\varphi}) ]  = \tilde{\mathcal{E}}_{\mathrm{G}}(\boldsymbol{\sigma}_A) \,, \\
\tilde{\mathcal{E}}_{\mathrm{G}}[\tr_A (\boldsymbol{\pi})] \geqslant \tilde{\mathcal{E}}_{\mathrm{G}}[\tr_A (\boldsymbol{\pi}+\boldsymbol{\varphi})]   = \tilde{\mathcal{E}}_{\mathrm{G}}(\boldsymbol{\sigma}_B) \,,
\end{gather} 
\end{subequations}
so we have 
\begin{equation}
\tilde{\mathcal{E}}_{\mathrm{G}}[\tr_A(\boldsymbol{\pi})] + \tilde{\mathcal{E}}_{\mathrm{G}}[\tr_B (\boldsymbol{\pi})]  \geqslant \tilde{\mathcal{E}}_{\mathrm{G}}(\boldsymbol{\sigma}_A)+\tilde{\mathcal{E}}_{\mathrm{G}}(\boldsymbol{\sigma}_B)   \,.
\label{ineq1}
\end{equation}
The $N'$-mode state $\tr_A (\boldsymbol{\pi})$ and $N$-mode state $\tr_B (\boldsymbol{\pi})$ in the above inequality can also decomposed as follows
\begin{subequations}
\begin{gather}
\tr_A (\boldsymbol{\pi}) = \boldsymbol{\pi}_{B} + \boldsymbol{\varphi}_B \,, \\
\tr_B (\boldsymbol{\pi}) = \boldsymbol{\pi}_{A} + \boldsymbol{\varphi}_A \,,
\end{gather}
\label{eqredstates}
\end{subequations}
and again for arbitrary $\boldsymbol{\varphi}_A \geqslant 0$ and $\boldsymbol{\varphi}_B \geqslant 0$ we have
\begin{subequations}
\begin{gather}
\tilde{\mathcal{E}}_{\mathrm{G}} (\boldsymbol{\pi}_{A}) \geqslant \tilde{\mathcal{E}}_{\mathrm{G}} (\boldsymbol{\pi}_{A} + \boldsymbol{\varphi}_A)  = \tilde{\mathcal{E}}_{\mathrm{G}} [ \tr_B( \boldsymbol{\pi})] \,, \\
\tilde{\mathcal{E}}_{\mathrm{G}} (\boldsymbol{\pi}_{B}) \geqslant \tilde{\mathcal{E}}_{\mathrm{G}} (\boldsymbol{\pi}_{B} + \boldsymbol{\varphi}_B)  = \tilde{\mathcal{E}}_{\mathrm{G}} [ \tr_A (\boldsymbol{\pi})]\,,
\end{gather} 
\end{subequations}
which implies
\begin{equation}
\tilde{\mathcal{E}}_{\mathrm{G}} (\boldsymbol{\pi}_{A}) + \tilde{\mathcal{E}}_{\mathrm{G}} (\boldsymbol{\pi}_{B}) \geqslant \tilde{\mathcal{E}}_{\mathrm{G}} [ \tr_A (\boldsymbol{\pi})] + \tilde{\mathcal{E}}_{\mathrm{G}} [ \tr_B (\boldsymbol{\pi})]   \,.
\label{ineq2}
\end{equation}
Since $\boldsymbol{\pi}_{A}$ and $\boldsymbol{\pi}_{B}$ are pure states, their NGEoF is equivalent to their entropy of entanglement, i.e.,
\begin{subequations}
\begin{gather}
{}\tilde{\mathcal{E}}_{\mathrm{G}} (\boldsymbol{\pi}_{A}) = \tilde{E} (\boldsymbol{\pi}_{A}) = \tilde{S} (\boldsymbol{\pi}_{A})\,,
\\
{}\tilde{\mathcal{E}}_{\mathrm{G}} (\boldsymbol{\pi}_{B}) = \tilde{E} (\boldsymbol{\pi}_{B}) = \tilde{S} (\boldsymbol{\pi}_{B})\,,
\end{gather} 
\end{subequations}
and for arbitrary $\boldsymbol{\varphi}_A \geqslant 0$ and $\boldsymbol{\varphi}_B \geqslant 0$ we get
\begin{subequations}
\begin{gather}
\tilde{S} [\tr_B (\boldsymbol{\pi})] = \tilde{S}(\boldsymbol{\pi}_{A} + \boldsymbol{\varphi}_A ) \geqslant \tilde{S}(\boldsymbol{\pi}_{A}) \,, \\ 
\tilde{S} [\tr_A (\boldsymbol{\pi})]= \tilde{S}(\boldsymbol{\pi}_{B} + \boldsymbol{\varphi}_B ) \geqslant \tilde{S}(\boldsymbol{\pi}_{B}) \,
\end{gather} 
\end{subequations}
which combined with the inequality (\ref{ineq1}) and (\ref{ineq2}) turns into 
\begin{equation}
\tilde{S} [\tr_A (\boldsymbol{\pi})] + \tilde{S} [\tr_B(\boldsymbol{\pi})] \geqslant \tilde{\mathcal{E}}_{\mathrm{G}}(\boldsymbol{\sigma}_A)+\tilde{\mathcal{E}}_{\mathrm{G}}(\boldsymbol{\sigma}_B) \,.
\label{ineq3}
\end{equation}
We now notice that for any $(N+N')$-mode state $\boldsymbol{\sigma}$ we have
\begin{equation}
\tilde{S}(\boldsymbol{\sigma}_{s_1s_2})= \tilde{S}[\tr_{s_1}(\boldsymbol{\sigma}_{s_1s_2})]+\tilde{S}[\tr_{s_2}(\boldsymbol{\sigma}_{s_1s_2})]
\end{equation}
and thus the left-hand side of the inequality (\ref{ineq3}) becomes
\begin{equation}
\tilde{S} [\tr_B (\boldsymbol{\pi})] + \tilde{S} [\tr_A (\boldsymbol{\pi})] = \tilde{S}(\boldsymbol{\pi})= \tilde{E}(\boldsymbol{\pi}) \,.
\label{eq1}
\end{equation}

Given that the above equality is true for every $\boldsymbol{\pi}$, it should be also true for the ``optimal'' $\boldsymbol{\pi}_{o}$ that gives the NGEoF of the global state $\boldsymbol{\sigma}=\boldsymbol{\sigma}_A \oplus \boldsymbol{\sigma}_B$ in Eqn.~(\ref{decomposition}), i.e., 
\begin{equation}
\tilde{\mathcal{E}}_{\mathrm{G}}(\boldsymbol{\sigma}) = \tilde{\mathcal{E}}_{\mathrm{G}}(\boldsymbol{\sigma}_A \oplus \boldsymbol{\sigma}_B) = \tilde{E}(\boldsymbol{\pi}_{o})\, .
\label{eq2}
\end{equation}

Combining the above Eqns.~(\ref{eq1}) and (\ref{eq2}) with the inequality (\ref{ineq3}), we get

\begin{equation}
\tilde{\mathcal{E}}_{\mathrm{G}}(\boldsymbol{\sigma}_A \oplus \boldsymbol{\sigma}_B) \geqslant \tilde{\mathcal{E}}_{\mathrm{G}}(\boldsymbol{\sigma}_A) + \tilde{\mathcal{E}}_{\mathrm{G}}(\boldsymbol{\sigma}_B)\, ,
\end{equation}
which completes the proof.
\end{proof}

\section{Gaussian Local Unitary Operations}\label{ApGLUO}

In this section we introduce a useful class of operations Gaussian local unitary operations (GLUO). GLUO are operations which do not increase or decrease the amount of entanglement. By definition, these operations are a subset of LOCC (here, locality is defined with respect to each mode), which means that they cannot increase the entanglement. As these operations are locally reversible (i.e. unitary in terms of the Heisenberg picture), they cannot decrease the entanglement.

\begin{figure}[t]\label{FigGLUO}
\includegraphics[width=0.6\textwidth]{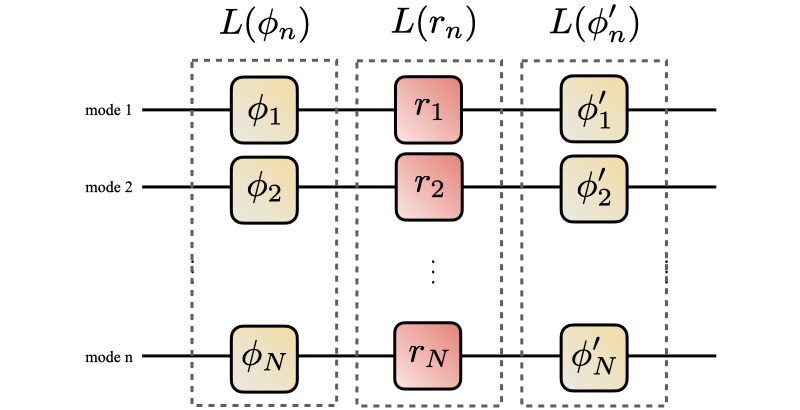}
\caption{A schematic decomposition of all GLUO operations.}
\label{fig: Unitary}
\end{figure}

We introduce the GLUO of a $N$-mode state as follows:
\begin{equation}
\boldsymbol{L}\equiv \bigoplus_{n=1}^N \; \boldsymbol{L}_n \,,
\end{equation}
where $\boldsymbol{L}_n$ is the GLUO in each mode. Each GLUO can be decomposed through the Bloch Messiah decomposition \cite{bloch1962canonical, Braunstein2005Squeezing} as
\begin{equation}
\boldsymbol{L}_n= \boldsymbol{L}(\phi_n')\boldsymbol{L}(r_n)\boldsymbol{L}(\phi_n) \,,
\end{equation}
where
\begin{equation}
\boldsymbol{L}(\phi) \equiv \begin{bmatrix}
\cos(\phi) & \sin(\phi) \\
\sin(\phi)  & \cos(\phi)  \\
\end{bmatrix} \,,
\end{equation}
corresponds to phase rotations, and
\begin{equation}
\boldsymbol{L}(r) \equiv \begin{bmatrix}
\cosh(r) & 0 \\
0 & \sinh(r)  \\
\end{bmatrix} \,,
\end{equation}
corresponds squeezing operations. A schematic diagram of this decomposition for GLUO is shown in Fig \ref{FigGLUO}.

\section{Optimization of NGEoF for q-p states}\label{ApQp}
\begin{proposition}
Consider a q-p state $\boldsymbol{\sigma}_{\text{qp}}$. For every pure state, $\boldsymbol{\pi} \leqslant \boldsymbol{\sigma}$, there exists a q-p pure state $\boldsymbol{\pi}_{\text{qp}}'\leqslant \boldsymbol{\sigma}$ which satisfies the following:
\begin{equation}
\tilde{E}(\boldsymbol{\pi})\geqslant \tilde{E}(\boldsymbol{\pi}'_{\text{qp}})
\end{equation}
\end{proposition}
\begin{proof} Any Gaussian pure state $\boldsymbol{\pi}$ can be written in the following way \cite{wolf2004gaussian}:
\begin{equation}
 \boldsymbol{\pi}(\boldsymbol{X},\boldsymbol{Y})=\begin{bmatrix}
\boldsymbol{X} & \boldsymbol{XY} \\
\boldsymbol{YX} & \boldsymbol{YXY}+\boldsymbol{X}^{-1}\\
\end{bmatrix} \;,
\end{equation}
where $\boldsymbol{X}>0$ and $\boldsymbol{Y}$ are real symmetric $N \times N$ matrix with $\boldsymbol{X}>0$. For q-p states, $\boldsymbol{Y}=0$. For every $\boldsymbol{\sigma}_{\text{qp}}\geqslant \boldsymbol{\pi}(\boldsymbol{X},\boldsymbol{Y})$,  we have the following \cite{wolf2004gaussian}:
\begin{equation}\label{Eqnpq1}
 \boldsymbol{\sigma}_{\text{qp}} \geqslant \boldsymbol{\pi}(\boldsymbol{X},\boldsymbol{Y}) \Rightarrow  \boldsymbol{\sigma}_{\text{qp}} \geqslant \boldsymbol{\pi}(\boldsymbol{X},0)
\end{equation}
We also have that the determinant of the single mode $\boldsymbol{\pi}_n(\boldsymbol{X},\boldsymbol{Y})$ is always larger than $\boldsymbol{\pi}_n(\boldsymbol{X},\boldsymbol{0})$:
\begin{equation}\label{EqDetXY}
\det\left[\boldsymbol{\pi}_n(\boldsymbol{X},\boldsymbol{0})\right]\leqslant \det\left[\boldsymbol{\pi}_n(\boldsymbol{X},\boldsymbol{Y})\right]
\end{equation}
The entropy of a single mode state is computed to be;
\begin{equation}\label{EqnSn}
S(\boldsymbol{\sigma}_n)=h\left[\sqrt{\det(\boldsymbol{\sigma}_n)}\right] \, ,
\end{equation}
As this is true for every mode, combining Eqn. (\ref{EqnSn}) and (\ref{EqDetXY}) gives the following:
\begin{equation}\label{Eqnpq2}
\tilde{E}[\boldsymbol{\pi}_n(\boldsymbol{X},\boldsymbol{0})] \leqslant \tilde{E}\left[\boldsymbol{\pi}_n[\boldsymbol{X},\boldsymbol{Y})\right]
\end{equation}
Eqn. (\ref{Eqnpq1}) and (\ref{Eqnpq2}) completes the proof.
\end{proof}
\section{Symmetric 3-Mode Squeezing Operation}\label{ApS3}

The Heisenberg evolution of a three mode squeezing operation is as follows \cite{Wu2005Continuous}:
\begin{equation}
\hat{S}_3^{\dag} \hat{a}_i\hat{S}_3= \cosh(r)\hat{a}_i+\sinh(r)\left[-\frac{1}{3}\hat{a}_i^{\dag}+\frac{2}{3}(\hat{a}_j^{\dag}+\hat{a}_k^{\dag})\right] \,.
\end{equation} 
The covariance matrix representation of a three mode squeezer is given by
\begin{equation}
\boldsymbol{{S}}_3(r_3)=\begin{bmatrix}
{\alpha}_+& {\beta}_+& {\beta}_+&0&0&0\\
{\beta}_+& {\alpha}_+& {\beta}_+&0&0&0\\
{\beta}_+& {\beta}_+& {\alpha}_+&0&0&0\\
0&0&0&{\alpha}_-& {\beta}_-& {\beta}_-\\
0&0&0&{\beta}_-& {\alpha}_-& {\beta}_-\\
0&0&0&{\beta}_-& {\beta}_-& {\alpha}_-\\
\end{bmatrix} \,,
\end{equation}
where we have defined the following:
\begin{equation}
\alpha_{\pm} \equiv 
\cosh(r_3)\mp\frac{\sinh(r_3)}{3}\,, \quad
\beta_{\pm} \equiv
\pm\frac{2\sinh(r_3)}{3}\,.
\end{equation}

We obtain the GhZ/W state \cite{adesso2006multipartite} when we apply this operator onto the vacuum state. In the standard form \cite{adesso2005gaussian}, this state can be written in the following way:
\begin{equation}
\boldsymbol{\pi}_{GhZ/W,\text{sf}}(r_3)\equiv(\boldsymbol{S}_3 \boldsymbol{S}_3^T)_{\text{sf}}= \;\begin{bmatrix}
{\alpha}'& {\beta}'_+& {\beta}'_+&0&0&0\\
{\beta}'_+& {\alpha}'& {\beta}'_+&0&0&0\\
{\beta}'_+& {\beta}'_+& {\alpha}'&0&0&0\\
0&0&0&{\alpha}'& {\beta}'_-& {\beta}'_-\\
0&0&0&{\beta}'_-& {\alpha}'& {\beta}'_-\\
0&0&0&{\beta}'_-& {\beta}'_-& {\alpha}'\\
\end{bmatrix} \,,
\end{equation}
where
\begin{equation}
\alpha' \equiv  \frac{1}{3}\sqrt{9\cosh(2r_3)^2-\sinh(2r_3)^2} \,, \quad 
\beta_{\pm} \equiv \pm \frac{|2\sinh(2r_3)|}{3} \sqrt{\frac{3\cosh(2r_3)\pm|\sinh(2r_3)|}{3\cosh(2r_3)\mp|\sinh(2r_3)|}} \,.
\end{equation}
The Bloch-Messiah decomposition \cite{bloch1962canonical, Braunstein2005Squeezing} of this operator can be found in a straightforward fashion by setting the local squeezers to be equal with $2 \pi/3$ phase differences.

\twocolumngrid
\bibliography{bibliography}
\end{document}